\documentclass[12pt,a4paper,reqno]{amsart}
\usepackage{amsmath,amssymb}
\usepackage{amsmath,amssymb}
\usepackage{latexsym}
\usepackage{pstricks}
\usepackage{pst-node}
\usepackage{pst-text}
\usepackage{graphicx}
\usepackage[T1]{fontenc}
\usepackage[english]{babel}

\newcommand{\lap}{\Delta}

\newcommand{\R}{\mathbb R}

\newcommand{\de}{\mathop{}\!\mathrm{d}}  
\newcommand{\pa}{\mathop{}\!\partial}

\usepackage{tikz}
\usetikzlibrary{calc}
\usepackage{tikz,pgfplots}
    \usetikzlibrary{pgfplots.polar,intersections}
    \pgfplotsset{compat=newest}
\usetikzlibrary{hobby,patterns}

\usepackage{comment}

\newcommand{\ve}[1]{{\boldsymbol {#1}}}
\newcommand{\fcr}{\chi}

\usepackage{float}

\newtheorem{conj}{Conjecture}
\newtheorem{remark}{Remark}
\newtheorem{theorem}{Theorem}
\newtheorem{lemma}{Lemma}

\numberwithin{equation}{section}
\begin{document}

\title
{Euclidean random matching in 2D for non-constant densities}

\author[D. Benedetto]{Dario Benedetto}
\address{Dario Benedetto \hfill\break \indent
        Dipartimento di Matematica, 
        Sapienza Universit\`a di Roma,
        \hfill\break \indent
        P.le Aldo Moro 5, 00185 Roma, Italy}
\email{benedetto@mat.uniroma1.it}
\author[E. Caglioti]{Emanuele Caglioti}
\address{Emanuele Caglioti \hfill\break \indent
        Dipartimento di Matematica, 
        Sapienza Universit\`a di Roma,
        \hfill\break \indent
        P.le Aldo Moro 5, 00185 Roma, Italy}
\email{caglioti@mat.uniroma1.it}

\begin{abstract}
  We consider the two-dimensional random matching problem in
  $\mathbb{R}^2.$ In a challenging paper, Caracciolo et al. \cite{CLPS},
  on the basis of a subtle  linearization of the Monge-Amp\`ere
  equation, conjectured that the expected value of the square of
  the Wasserstein
  distance, with exponent $2,$ between two samples of $N$ uniformly
  distributed points in the unit square is $\log N/2\pi N$ plus
  corrections, while the expected value of the square of the
  Wasserstein distance
  between one sample of $N$ uniformly distributed points and the uniform
  measure on the square is $\log N/4\pi N$.
  These conjectures have been proved by Ambrosio et al. \cite{AST}.
        
  Here we consider the case in which the points are sampled from a
  non-uniform density.
  For first we give formal arguments leading to the conjecture that if
  the density is regular and positive in a regular, bounded and
  connected domain
  $\Lambda$ in the plane, then the leading term of the expected values
  of the Wasserstein distances are exactly the same as in the
  case of uniform density, but for the
  multiplicative factor equal to the measure of $\Lambda$.

  We do not prove these results but, in the case in which the
  domain is a square, we prove estimates from above that
  coincides with the conjectured result.

\end{abstract}


\maketitle

\section{Introduction}
\label{sez:intro}

Let $\mu$ be a probability distribution defined on the unit square
$Q=[0,1]^2.$ Let us consider two sets $\underline {x}^N
=\{\ve x_i\}_{i=1}^N$ and 
$\underline {y}^N = \{\ve y_i\}_{i=1}^N$
of $N$ points independently sampled from the distribution $\mu$.
The Euclidean Matching problem with exponent $2$ consists in finding
the matching $i\to \pi_i$, i.e. the permutation $\pi$
of $\{1,\dots N\}$ which minimizes
the sum of the squares of the distances between $\ve x_i$ and
$\ve y_{\pi_i}$,
that is
\begin{equation}
  \label{eq:cn}
  C_N(\underline {x}^N, \underline {y}^N)= \min_{\pi}\sum_{i=1}^N
  |\ve x_i-\ve y_{\pi_i}|^2.
\end{equation}
The cost defined above can be seen, but for a
constant factor $N$, as the square of the 2-Wasserstein distance
between two probability measures.
In fact, the $p-$Wasserstein distance $W_p(\nu_1,\nu_2)$,
with exponent $p\geq1$, between two probability
measures $\nu_1$ and  $\nu_2,$ is defined by 
\begin{equation*}
  W_p^p(\nu_1,\nu_2)=\inf_{J_{\nu_1,\nu_2}}
  \int J_{\nu_1, \nu_2}
  (\de \ve x,\de \ve y) |\ve x-\ve y|^p,
\end{equation*}
where the infimum is taken on all the joint probability
distributions $J_{\nu_1,\nu_1}(\de \ve x, \de \ve y)$
with marginals with respect to $\de \ve x$ and $\de \ve y$ given by
$\nu_1(\de \ve x)$ and $\nu_2(\de \ve y)$, respectively.  
Defining the empirical measures
$$
X^N(\de \ve x)=\frac{1}{N}\sum_{i=1}^N\delta_{\ve x_i}(\ve x)
\de \ve x,\phantom{aa}
Y^N(\de \ve x)=\frac{1}{N}\sum_{i=1}^N\delta_{\ve y_i}(\ve x)
\de \ve x,
$$
it is possible to show that
$$C_N(\underline {x}^N, \underline {y}^N) = N W^2_2(X^N,Y^N),$$
(see for instance \cite{B}).
In the sequel we will shorten
$C_N = C_N(\underline {x}^N, \underline {y}^N)$.

In the  challenging paper
\cite{CLPS},
at first for particles in the torus of measure one,
then also in the case of the square, see \cite{CS}, 
Caracciolo et al. 
conjectured that when  $\ve x_i$ and $\ve y_i$ are
sampled independently with  uniform density on $Q$,
then
\begin{equation}
  \label{eq:ecn}
  \mathbb {E}_\sigma[C_N] \sim \frac{\log N}{2\pi},
\end{equation}
where with $\mathbb {E}_\sigma$ we denoted the expected value with
respect to the uniform distribution $\sigma(\de \ve x) = \de \ve x$
of the points $\{\ve x_i\}$ and $\{\ve y_i\}$, 
and where we say that $f\sim g$
if $\lim_{N\rightarrow+\infty} f(N)/g(N) =1$.
In terms of $W_2^2$ the conjecture is equivalent to 
\begin{equation}
  \label{fact1}
  \mathbb{E}_\sigma[W_2^2(X^N,Y^N)] \sim \frac{\log N}{2\pi N}.
\end{equation}

Moreover, in \cite{CLPS} it is conjectured that asymptotic of the
expected value of
$W_2^2(X^N,\sigma)$ between the empirical density $X^N$ and the uniform
probability measure $\sigma(\de \ve x)$ on $Q$ is given by
\begin{equation}
  \label{fact2}
  \mathbb{E}_\sigma[W_2^2(X^N,\sigma)]
  \sim \frac{\log N}{4\pi N}.
\end{equation}
A first general results showing that in the case of the unit square
$\mathbb{E}_\sigma[W_2^2(X^N,Y^N)] $
behaves as $\frac{\log N}{N}$ has been
obtained in \cite{AKT}.
The conjectures above has been proved by Ambrosio et al. \cite{AST}. 
In \cite{AG} finer estimates are given and it is proved
that the result can be extended to the case when
the particles are sampled from the volume measure on a two-dimensional
Riemannian compact manifold. In \cite{AGT} it is shown that the
properties of the optimal transport map
for  $W_2(X^N,\sigma)$ are in
agreement with the result in \cite{CLPS}.

\vskip.3cm

We notice that if we consider square (or manifold) of measure $|Q|\neq 1$,
the cost has to be multiplied by $|Q|$.
Namely, if we extract $\{\ve x_i\}$
uniformly in $Q$, then the points $\{\gamma \ve x_i\}$, with $\gamma >0$,
are uniformly distributed in $\gamma Q$, and
$C_N(\underline {x}^N, \underline {y}^N) = \gamma^{-2}
C_N(\gamma \underline {x}^N, \gamma \underline {y}^N)$.
By imposing that $|\gamma Q| = 1$, i.e. 
$\gamma^{-2} = |Q|$,  we obtain that the expectation of the cost
$C_N(\gamma \underline {x}^N, \gamma \underline {y}^N)$
verifies the asymptotic
estimate \eqref{eq:ecn}.

\vskip.3cm In this paper we consider the case of non-uniform measure
$\mu(\de \ve x) = \rho(\ve x)\de \ve x$
with $\rho$ strictly positive and regular.

In particular in Section \ref{sez:due} we 
study the asymptotic behavior of the expected value of the cost 
when  $\rho$ is a density on $Q$, piecewise constant on a grid of sub-squares.
On the basis of the 
analysis of this case,  in Conjecture \ref{conj:1-2} we guess that,
 in the case of regular and strictly positive density, 
the asymptotic 
behavior is still described by the right-hand-sides of
eq.s \eqref{fact1} and
\eqref{fact2}.
In the case of a density defined on
a regular connected bounded set $\Lambda$ in the plane,
we expect that the asymptotic behavior changes
only for the multiplicative factor
$|\Lambda|$ (see Conjecture  \ref{conj:3-4}).

In Section \ref{sez:3nuova} we face the random Euclidean matching problem
with the strategy presented in \cite{CLPS,CS}, showing that the results
conjectured in \ref{sez:due}
can be formally justified on the basis of that approach.

We do not fully prove the conjectures, but in section \ref{sez:stima}
we prove that
\eqref{fact1} and \eqref{fact2} give exact estimates from above of the
cost, in the case of strictly positive and Lipschitz continuous
density on $Q$.

\vskip.3cm
  
Although this work concerns the two-dimensional case for 
cost and Wasserstein distance with exponent $2$, we briefly review here what
is known in the other cases, up to our knowledge. 

In dimension $2$, for $p\ge 1$, $p\neq 2$, in \cite{AKT}
it has been proved that the expected cost per particle $\mathbb{E}[C_N]/N$
is
$O(N^{-p/2})$ as $N\rightarrow\infty.$
The value of the limit as $N\to +\infty$ of 
$N^{p/2}\mathbb{E}[C_N]/N$ is not known. 

In dimension $1$ the random Euclidean matching problem in a segment is
almost completely characterized, for any $p\geq 1$. This is due to the
fact that the best matching between two set of points on a line is 
monotone.
When the density is uniform on the segment $[0,1]$ and $p=2$, one gets
$\mathbb{E}[C_N]\to 1/3$ as $N\to\infty$. 
In this case, it is also possible  to compute explicitly
$\mathbb{E}[C_N]$ for any $N$, in fact
$\mathbb{E}[C_N]={N}/{3 (N +1)}$.
Moreover, for  any $p\geq 1$,
$\mathbb{E}[C_N]/N\sim cN^{-p/2}$,
and a general expression for $c$  has been determined in \cite{CS2014}.
The different behavior in the case in which the density vanishes in
some point or in a segment has been analyzed in \cite{CDS}
(see the Remark 3 at the end of Section 2).
A general discussion on the one-dimensional, also for the case of
non-constant densities
and of densities defined on all the line, can be found in 
\cite{BL}. 

In dimension $d\geq 3$, it has been proved that
$\mathbb{E}[C_N]/N\sim cN^{-p/d}$, for any $p\geq 1$ (see \cite{Ta92},
\cite{DY} for $p=1,$ and \cite{Le} for $p\geq 1$).

\section{Some conjectures for non-constant densities}
\label{sez:due}

Let us consider the case 
$\mu(\de \ve x)=\rho(\ve x)\,\de \ve x$ with 
$\rho(\ve x)$ is piecewise constant with respect to a
regular grid of sub-squares of $Q$.
For sake of simplicity we consider the case in which the grid is made
by four sub-squares: $[0,1/2)^2$, 
$[0,1/2)\times[1/2,1],$ $[1/2,1]\times[0,1/2),$  $[1/2,1]^2$
(see fig. \ref{fig:4-quadrati}).

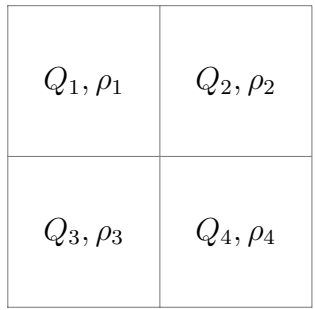
\begin{figure}[hbt!]
\centering
\begin{tikzpicture}
\draw[step=2.cm,color=gray] (-2.01,-2.01) grid (2,2);
\node at (-1.,+1) {$Q_1,\rho_1$};

\node at (-1,-1) {$Q_3,\rho_3$};
\node at (1,1) {$Q_2,\rho_2$};
\node at (+1,-1) {$Q_4,\rho_4$};

\end{tikzpicture}
\caption{Grid of $2\times2$ squares.}
\label{fig:4-quadrati}
\end{figure}
Let us denote by $Q_k:i=1,\dots 4,$ the four squares and by
$\rho_k > 0$, $k=1,\dots 4$ the corresponding constant densities.
Now, let $\{\ve x_i\}_{i=1}^N$ and $\{\ve y_i\}_{i=1}^N$
be two samples of $N$ independent
points from the distribution $\mu$, and let us denote
with $R_k$ and $S_k$ the number of points $\ve x_i$ and $\ve y_i$
in $Q_k$,
respectively.  Then, both $R_k$ and $S_k$ will be equal to
$N_k=\rho_kN/4$ plus terms of the order of $\sqrt{N}$.

Now we make two ansatzes.
\begin{enumerate}
\item Up to a correction $o(\log N)$,
  we can calculate $\mathbb E_\mu[C_N]$ by 
  restricting ourselves to the case in which
  both $R_k$ and $S_k$ are
  equal to $N_k=\rho_kN/4$
  (rounded to integer numbers  in such a way that the sum of the $N_k$ is $N$).
  
\item Given the samples with $R_k=S_k=N_k$,
  the optimal cost,
  with the constraint that $\ve x_i$ and $\ve y_{\pi_i}$
  are in the same square, is $C_N$ plus an error $o(\log N)$.
\end{enumerate}
Under these assumptions we get that, but for terms of order $1$, the
expected value of the cost of the optimal matching will be given by
the sum of the expected value of the cost of the optimal couplings in
the four squares.

Now let us notice that, by eq. \eqref{fact1}, if we sample $N_k$ particles
uniformly and independently in a square of size
$|Q_k|$, then the expected value of the cost is simply given
by $|Q_k|\frac{\log N}{2\pi}$, as follows
by the scaling argument shown in the previous section.

Therefore,
\begin{equation*}
  \begin{split}
    \mathbb{E}_\mu[C_N]
    &=\sum_{k=1}^4 |Q_k|
    \frac{\log (\rho_k  N/4)}{2\pi}+o(\log N)\\
    &=\sum_{k=1}^4 |Q_k|
    \frac{\log N}{2\pi}+\sum_{k=1}^4 |Q_k|
    \frac{\log (\rho_k/4)}{2\pi}+o(\log N)\\
    &=
    \frac{\log N}{2\pi}\left(\sum_{k=1}^4 |Q_k|\right)+o(\log N)=
    \frac{\log N}{2\pi}+o(\log N),
  \end{split}
\end{equation*}
where we used that $\sum |Q_k| =1$.
We can notice that the dependence of $\mathbb{E}_\mu[C_N]$ on the values
of the densities $\rho_k$ does not affect the leading term, that
only depends on the measure of the set.

This analysis can be extended when we consider a regular grid of $m^2$
squares.  Therefore, by noticing that it is possible to approximate a
continuous density $\rho$ as well as we want in $L_{\infty}$ with a
piecewise constant density, we are led to the following conjectures.

\begin{conj}
  \label{conj:1-2}
  Let  $\mu(\de \ve x)  = \rho(\ve x)\de \ve x$
  a probability measure defined on $Q$ where $\rho$ is a smooth
  positive density on $Q.$
   Let  $\{\ve x_i\}_{i=1}^N$
  and $\{\ve y_i\}_{i=1}^N$ be
  two samples of points independently distributed
  with $\mu$.
  Then 
  \begin{equation}
    \label{eq:conj11}
    \mathbb{E}_\mu[C_N]  \sim\frac{\log N}{2\pi}.
  \end{equation}
Reasoning in the same way,
we can conjecture that the asymptotic behavior of the $2-$Wasserstein distance
between the empirical measure
$X^N$  and the measure $\mu$ itself verifies
\begin{equation}
  \label{eq:conj12}
    \mathbb{E}_\mu[W_2^2(X^N,\mu)] \sim
    \frac{\log N}{4 \pi N}.
  \end{equation}
\end{conj}

Let us notice that the two ansatzes above are far from been
obvious.  Nevertheless, in the next section we will prove 
that the right-hand-sides of eq.s \eqref{eq:conj11} and \eqref{eq:conj12}
give exact estimates from
above of the expected values.

\vskip.3cm
Let us now consider a bounded connected set $\Lambda$ in
$\mathbb{R}^2$ with regular boundary, and consider a partition of
$\Lambda$ with squares of sides $1/m$, as in fig. \ref{fig:lambda}.
Let us suppose that the probability measure $\mu$ has a 
smooth and positive density in $\Lambda$,
and define $\Lambda_k = Q_K \cap \Lambda$.

\begin{figure}[hbt!]
\centering
\begin{tikzpicture}
 \coordinate (N) at (8.5,20);
 \coordinate (O) at (5.9,11);
 \coordinate (P) at (6.6,10.8);
 \coordinate (Q) at (7.6,12);
 \coordinate (R) at (8.4,12.6);
 \coordinate (S) at (8.1,14);
 \coordinate (T) at (7.4,13.8);
 \coordinate [label={[xshift=-0.4cm, yshift=0.4cm]$\Lambda$}](U) at (6.4,13.2);
 \draw[pattern=grid] (O) to [ closed, curve through = {(P)  (Q)  (R) (S)  (T) (U)}] (O);
\end{tikzpicture}
\caption{Set $\Lambda$ covered with squares.}
\label{fig:lambda}
\end{figure}
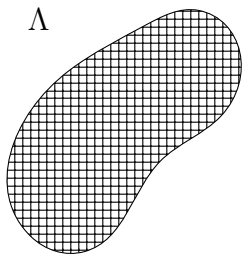

Then with the same reasoning made for the case
of the square $Q$, formally we get
\begin{equation*}
  \begin{split}
    \mathbb{E}_\mu[C_N]&\sim
    \sum_{k:Q_k \subset \Lambda}
    |Q_k|\frac{\log (\rho_k  N/|Q_k|)}{2\pi}\\
    &=\frac{\log N}{2\pi}\sum_{k:Q_k \subset \Lambda}
    |Q_k|+O(1)\sim
    |\Lambda|\frac{\log N}{2\pi},
  \end{split}
\end{equation*}
where $\rho_k$ is the average of $\rho$ on $\Lambda_k$.  In fact,
we expect that
any of the square $Q_k$ in $\Lambda$ contributes to $\mathbb{E}[C_N]$
with a term $ \sim\frac{|Q_k|}{2\pi}\log N $.
We have also
neglected the contribution of the squares close to the boundary.

Therefore, we are led to the following conjecture.
\begin{conj}
  \label{conj:3-4}
  Let  $\mu(\de \ve x)  = \rho(\ve x)\de \ve x$, a probability
  measure defined on $\Lambda$ where $\rho$ is a smooth positive density.
  Let $\{\ve x_i\}_{i=1}^N$ and $\{\ve y_i\}_{i=1}^N$
  two samples of $N$ points independently distributed with $\mu$.
  \begin{equation*}
    \mathbb{E_\mu}[C_N]
    \sim |\Lambda|\frac{\log N}{2\pi}\ \ \text{ and } \ \  
    \mathbb{E_\mu}[W_2^2(X^N,\mu)] \sim |\Lambda|\frac{\log N}{4 \pi N}.
  \end{equation*}
\end{conj}

\begin{remark}
  If the measure of the support of $\mu$
  is infinite (for instance if the support is all $\mathbb{R}^2$), 
  we expect that 
  \begin{equation*}
    \lim_{N\rightarrow\infty}
    \frac{\mathbb{E}_\mu[C_N]}{\log N} =
    +\infty.
  \end{equation*}
  This is in agreement with the fact, proved by Talagrand in
  \cite{Ta}, that when the density is the Gaussian,
  i.e. $\rho=\frac{1}{2\pi}e^{-|x|^2/2},$ the average of the cost
  satisfies for large $N$
  $$ (\log N)^2 \leq \mathbb{E}_\mu[C_N]\leq C  (\log N)^2.$$
  Notice that an estimate from above proportional to $(\log N)^2$ was
  previously proved by Ledoux in \cite{Le}.
  Moreover, in \cite{Ta} the
  author says that a similar estimate can be obtained for densities
  $\rho \propto e^{-|x|^{\alpha}}$ obtaining a bound form below for
  the cost proportional to $(\log N)^{1+2/\alpha},$ and therefore much
  larger than $\log N$.
\end{remark}

\begin{remark}
  \label{rem:delta}
  In the above conjectures we require that $\rho$ is positive, but we can
  reformulate the conjectures using the measure of the
  support of $\rho$ instead of 
  the measure of $\Lambda$.
    The condition which really can change the asymptotic
  behavior of the cost is the connection of the support of $\rho$.
  Namely, if this condition is not satisfied, the result
  may be false. In particular if $\rho$ is constant in two squares
  whose distance is positive,
  we get that the expected value of cost is $O(\sqrt N) \gg O(\log N)$. To get  an idea of what happens, consider
  $$\mu(\de \ve x)
  =\frac 12 \left(
    \delta_{\ve z_1}(\ve x) + \delta_{\ve z_2}(\ve x)\right)
  \de \ve x.$$
  Then
  $$
  \begin{aligned}
    &X^N(\de \ve x) = \left( \frac RN \delta_{\ve z_1} (\ve x) + 
      \frac {N-R}N \delta_{\ve z_2} (\ve x) \right)\de \ve x,\\
    &Y^N(\de \ve x) = \left( \frac SN \delta_{\ve z_1} (\ve x) + 
      \frac {N-S}N \delta_{\ve z_2} (\ve x) \right)\de \ve x,\\
  \end{aligned}
  $$
  where $R$ and $S$ are independent binomial variables of mean $N/2$ and
  variance $N/4$.
  It is easy to show that
  $$C_N  = L^2\left|R - S\right|$$
  where  $L=|\ve z_1 - \ve z_2|$.
  Then, by noticing that $R-S$ has variance $N/2$,  by
  the  Central Limit Theorem
  we get that the leading term of the expected value of the cost
  is $L^2\sqrt{N/\pi}$.
  This behavior is independent of the dimension.
  The reader can find in the paper \cite{CDS}
  the exact asymptotic value in the one-dimensional case 
  of two disjoint intervals of the same length and
  with constant density. 
\end{remark}

  \begin{remark}
    
    Non-constant densities
    have been previously addressed
    in \cite{CS}, in which the authors present a general expression 
    which also allows the explicit calculation
    of the asymptotic value of the cost
    in the one dimensional case.
    Again on the one-dimensional
    case, 
    in \cite{CDS} the authors consider also a matching problem
    as in eq. \eqref{eq:cn}
    but where the distance appears with the power $p\ge 1$,
    not necessarily $2$.
    The expected value of the cost per particle
    $\mathbb{E}[C_N]/N$
    goes as  $c\,N^{-p/2}$, 
    but interestingly
    $c$ can diverge if the density  approaches zero at some point.
\end{remark}

\section{A formal proof}
\label{sez:3nuova}

It is possible to extend the method by Caracciolo et al. \cite{CLPS}
to a generic (positive) density. In particular in \cite{CS} a formula
for $\mathbb{E}_\mu[C_N]$
and for its fluctuations is presented, in the general case. The
formula for $\mathbb{E}_\mu[C_N]$ 
is 
computed 
in the case of the uniform
density $\sigma$ in the square, recovering the results in \cite{CLPS}.
Here we follow the approach presented in the papers above, considering
the general case of a smooth and positive density and deriving
formally eq. \eqref{eq:conj11} of Conjecture \ref{conj:1-2}
(eq. \eqref{eq:conj12} can be derived essentially in the same way).

In the framework of this approach, the
main argument we use to derive eq. \eqref{eq:conj11} consists in noticing
that the singular part of the Green function of the linearized
Monge-Amp\`ere equation, that in the case of a generic density is an
elliptic operator in divergence form, has a very simple expression.

\subsection{Constant density}
\label{sub:costante}

The strategy proposed in \cite{CLPS} to compute the expected value of
$C_N$ consists in linearizing the Monge-Amp\`ere equation
(which is the Euler Lagrange equation for the Monge-Kantorovich
problem) and then to put a suitable cut-off on the expression
founded. For first, we here report the argument in \cite{CLPS} for 
the case of constant density, and we
refer to \cite{CLPS,CS} for the justification of the approach
and further details.

By linearizing the Monge-Amp\`ere equation around the
uniform probability measure $\sigma(\de \ve x) = \de \ve x$,
the Wasserstein distance
between two regular measures is approximated by
\begin{equation}
  \label{eq:linear}
  \int_Q  |\nabla \psi |^2 = -\int_Q \psi \, \lap \psi 
\end{equation}
where $\psi$ solves
\begin{equation}
  \label{eq:psi} 
  \lap \psi =-\delta\rho,
\end{equation}
whit Neumann boundary conditions,
and where $\delta\rho$ is the difference of the densities
of the two measures.
We use formally \eqref{eq:linear} in the case of singular measures,
introducing later a suitable cut-off that make finite
the cost.
In the bipartite case
\begin{equation}
  \label{eq:dro}
  \delta\rho(\ve x) \de \ve x = X^N(\de \ve x) - Y^N(\de \ve x)
\end{equation}
and the cost is $N$ times the Wasserstein distance, that is
\begin{equation*}
  C_N\sim N\int_Q
 |\nabla \psi |^2=N \int_Q
 \psi  \,\delta\rho. 
\end{equation*}
It is convenient to introduce 
the Green function  $\phi_{\ve z}$ for the
Laplace problem on $Q$, 
which is the solution, with zero average,  of
$$\lap \phi_{\ve z}(\ve x) = - \delta_{\ve z}(\ve x) + 1,$$
with Neumann boundary conditions.
Since $\psi$ solves eq. \eqref{eq:psi} with
$\delta \rho$ given in eq. \eqref{eq:dro}, 
from the  definition of $\phi_{\ve z}(\ve x)$ we get
$$\psi(\ve x) = \lap^{-1} \delta \rho (\ve x) = \frac 1N
\left( \sum_{i=1}^N \phi_{\ve x_i}(\ve x)  -
  \sum_{j=1}^N \phi_{\ve y_j}(\ve x) \right)$$
and then 
\begin{equation*}
  C_N \sim \frac{1}{N}\int_Q
  \left(\sum_{i=1}^N\phi_{\ve x_i}
    -\sum_{j=1}^N
    \phi_{\ve x_j}\right)
  \left(\sum_{i=1}^N\delta_{\ve x_i}-
    \sum_{j=1}^N\delta_{\ve y_j}\right).
\end{equation*}
Taking the expectation in the location of the delta functions, and
using that the Green function has zero average, 
we get
\begin{equation*}
\begin{split}
  \mathbb{E}_\sigma[C_N]
  &\sim \frac 1N \, \mathbb{E}_\sigma\int_Q \left[\left(\sum_{i=1}^N
      \phi_{\ve x_i}-
      \sum_{j=1}^N\phi_{\ve y_j}\right)
    \left(\sum_{i=1}^N\delta_{\ve x_i}- \sum_{j=1}^N
      \delta_{\ve y_j}\right)\right]\\
  &= \frac{1}{N}\mathbb E_\sigma
  \sum_{i=1}^N\int_Q
  \left(\phi_{\ve x_i}
    \delta_{\ve x_i}+
    \phi_{\ve y_i}\delta_{\ve y_i}\right)=2
  \int_Q \de \ve z \int_Q
  \de \ve x \, |\nabla\phi_{\ve z}(\ve x)|^2 \\
  &=
  2 \int_Q
  \de \ve x \, |\nabla\phi_{\ve 0}(\ve x)|^2,
\end{split}
\end{equation*}
(the integral in $\ve x$ does not depend on the position of
$\ve z$, then we can fix it in  $\ve z = \ve 0$).
By Parseval's Lemma,  the right-hand-side can
be written in Fourier series, with respect to the base of cosines,  as
$$
\frac{2}{\pi^2} 
\sum_{\ve k\in{\mathbb N}^2\backslash
  \{{\ve 0\}}}
\frac{1}{|\ve k|^2}.$$
This series is not summable but a natural cut-off
can be imposed by summing up to $\ve k$ as large as
$\frac{1}{\lambda}$ where
$\lambda=\frac{1}{\sqrt{N}}$ is the characteristic length of the
system, i.e. the typical distance between a point $x$ and its closest
point $y$.
In this way one gets $\mathbb E_\sigma [C_N]
\sim \frac{1}{2\pi} \log N + O(1).$
It is important to notice that if the cut-off is chosen to be
$\lambda = \alpha/\sqrt{N}$ then the leading term of does not
depend on the constant $\alpha$, which  only affects the $O(1)$ term.

\vskip.3cm In order to face the case of a non-constant density, it is
convenient to make the previous computation in the position space, in which
the cut-off can be obtained by smoothing the delta-function evolving
it, with the heat semigroup, for a time $t=1/N$.  
We recall that the Green function can be written as
$$\phi_{\ve z}(\ve x) =-\frac{1}{2\pi}\log | \ve x - \ve z |
+ \gamma(\ve x,\ve z),$$ where
$\gamma$ is a regular function. We indicate whit
$f^t$ the evolution of a function $f$ with the
heat semigroup until the time $t$,  and with $G_t(\ve x)$
the heat kernel in the whole space $\R^2$.
We get again
\begin{equation}
  \label{eq:cg}
  \begin{aligned}
    \mathbb E_\sigma[C_N]&\sim  2\int_Q
    \phi_{\ve 0}^{t}(\ve x)\,
    \delta^{t}_{\ve 0}(\ve x) \de \ve x =
    2\int_Q \phi_{\ve 0}(\ve x) \,
    \delta^{2t}_{\ve 0}(\ve x) \de \ve x \\
    & =  -2\frac 1{2\pi}
    \int_\R \log |\ve x| \, G_{2t}(\ve x) 
    \de x + O(1)
    = -\frac 1{\pi} \log \sqrt{t} + O(1) \\
    &= \frac 1{2\pi} \log N + O(1).
  \end{aligned}
\end{equation}

\subsection{Non-constant density}
\label{sub:noncostante}

Now let us consider the case of a probability measure $\mu$ of 
positive and regular density
$\rho$. The main difference from the case of a constant density is
that the linearized Monge-Amp\`ere equation reads as
\begin{equation}
  \label{divergenceformeq}
  \nabla\cdot(\rho\nabla \psi)=-\delta\rho.
\end{equation}
(see for instance
\cite{Sant} and references therein).
Also in this case
$$C_N\sim N \int_Q
\rho \, | \nabla\psi |^2= N\int_Q
\psi \delta \rho
$$
where $\psi$ satisfies \eqref{divergenceformeq}.
We then introduce the Green function $\phi_{\ve z}(\ve x)$
which is the solution of  
\begin{equation}
  \label{div2}
  \nabla\cdot(\rho\nabla \phi_{\ve z})
  =-(\delta_{\ve z}-\rho)\ \ with \ \ \int_Q
  \phi_{\ve z}(\ve x) \rho(\ve x) \de \ve x = 0,
\end{equation}
getting 
$$\psi=\frac{1}{N} \sum_{i=1}^N\phi_{\ve x_i}-
\frac 1N\sum_{j=1}^N\phi_{\ve y_j}$$
and
\begin{equation*}
  C_N \sim \frac{1}{N}\int_Q
  \left(\sum_{i=1}^N\phi_{\ve x_i}-
    \sum_{j=1}^N\phi_{\ve y_j}\right) \left(
    \sum_{i=1}^N\delta_{\ve x_i}-
    \sum_{j=1}^N\delta_{\ve y_j}\right).
\end{equation*}
Taking the expectation in the location of the delta functions, that
are distributed with density $\rho$, we get
\begin{equation*}
  \begin{split}
    \mathbb{E}_\mu[C_N]
    &\sim \frac 1N 
    \mathbb{E}_\mu\int_Q
    \left(\sum_{i=1}^N\phi_{\ve x_i}-
      \sum_{j=1}^N\phi_{\ve y_j}\right) \left(\sum_{i=1}^N
      \delta_{\ve x_i}-
      \sum_{j=1}^N\delta_{\ve y_j}\right)\\
    &=\frac 1N
    \mathbb{E}_\mu\int_Q
    \sum_{i=1}^N\left(\phi_{\ve x_i}
      \delta_{\ve x_i}+\phi_{\ve y_i}
      \delta_{\ve y_i}\right) =
    2\int_Q\de \ve z \,\rho(\ve z)
    \int_Q\de \ve x\,
    \phi_{\ve z}(\ve x)
    \delta_{\ve z}
    (\ve x).
  \end{split}
\end{equation*}
The key observation we make here consists in noticing that in the
equation \eqref{div2}, that we rewrite as
$$\rho \lap \phi_{\ve z}  + \nabla\rho \cdot \nabla\phi_{\ve z}
= -\delta_{\ve z} + \rho,$$
the term $\nabla\rho \cdot \nabla\phi_{\ve z}$
is less  singular than the $\delta$  function, therefore
\begin{equation}
  \label{idea}
  \phi_{\ve z} (\ve x) =
  -\frac{1}{2\pi\rho(\ve z) }\log|\ve x-\ve z| + O(1)
\end{equation}
as $\ve x\rightarrow \ve z$
(see the Remark \ref{rem:3} at the end of this section).
Finally, we apply the cut-off by evolving 
$\delta_{\ve z}$ until
the time $t=1/N$ with the heat semigroup. Proceeding as
in eq. \eqref{eq:cg} 
$$
\begin{aligned}
\mathbb E_\mu [C_N ]&\sim  -\frac 1\pi\int_Q
\de \ve z \, \rho(\ve z)\, \int_Q\de \ve x
\frac 1{\rho(\ve z)} \log |\ve x - \ve z| G_{2t}(\ve x -\ve z) + O(1)\\
&=\frac{1}{2\pi}\log N\left(\int_Q \de \ve z
\right) +O(1) =\frac{1}{2\pi}\log N + O(1),
\end{aligned}
$$
that is in agreement with our conjecture.

{The argument can be generalized to any regular bounded domain $\Lambda$
in the plane and to the case of the torus. In the latter case, the 
operator $\lap$ requires periodic boundary conditions.

Indeed, changing the domain or the boundary condition only affects the regular part of the Green function in \eqref{idea}.}

\begin{remark}
  \label{rem:3}
  Denoting
  with $\Delta^{-1}$ the inverse of the Laplacian,
  we have
  $$\phi_{\ve z} = -\Delta^{-1} \left(
    \frac{\delta_{\ve z}-\rho}
  {\rho}\right) -\Delta^{-1}\frac{\nabla\rho \cdot
    \nabla\phi_{\ve z}}{\rho}.$$
  This expression suggests that divergent part of $\phi_{\ve z}$
  is
  $$-\frac 1{2\pi\rho(\ve z)} \log |\ve x - \ve z|,$$
  and then that 
  $|\nabla \rho \cdot \nabla \phi_{\ve z} / \rho|$
  is bounded by  $\frac{c}{|\ve x-\ve z|}$.
  It is easy to show that applying $\lap^{-1}$ to this term we obtain
  a bounded continuous function. A rigorous proof of \eqref{idea} when the domain is all $\mathbb{R}^2$ can
  be found, for instance, in \cite{BaOr}, and can be extended to
  the case of the square with minor modifications.
\end{remark}

\section{Estimate from above}
\label{sez:stima}

In this section we prove that 
\begin{theorem}
  \label{teo:stima}
  Let  $\mu(\de \ve x)  = \rho(\ve x)\de \ve x$ be  a probability measure
  defined on $Q$, where $\rho$ is a Lipschitz continuous
  strictly positive density.
  \begin{enumerate}
  \item
    Let $\{\ve x_i\}_{i=1}^N$
    and $\{\ve y_i\}_{i=1}^N$ be two samples 
    of $N$  points chosen 
    independently with  distribution $\mu$. 
    Then 
    \begin{equation}
      \limsup_{N\rightarrow\infty}
      \frac{2 \pi}{\log N}{\mathbb{E}_\mu[C_N]}
      \leq 1\label{estimate1} \end{equation}
    that is equivalent to 
    \begin{equation}\limsup_{N\rightarrow\infty}
      \frac{2 \pi N}{\log N}
      {\mathbb{E}_\mu[W_2^2(X^N,Y^N)]}\leq 1\label{estimate1bis}
    \end{equation}
  \item
    Moreover
    \begin{equation}
      \limsup_{N\rightarrow\infty}\frac{4 \pi N}{\log N}
             {\mathbb{E}_\mu[W_2^2(X^N,\mu)]}\leq 1
             \label{estimate2} \end{equation}
  \end{enumerate}
\end{theorem}

We first prove the second part of the theorem,
and then we show that \eqref{estimate2} implies \eqref{estimate1bis}.

\vskip.3cm
The idea of the proof is to  divide the square $Q$ into small squares
where the density can be considered constant in order to apply
the result in eq. \eqref{fact2}.
More precisely, we state the following Lemma.

\begin{lemma}
  \label{lemma:ell}
  Let $\rho(\ve x)$ be a strictly positive and Lipschitz continuous function
  defined  in 
  $Q^\ell = [0,\ell]^2$,
  let $\nu(\de \ve x) = r(\ve x) \de x$ be the probability
  measure of density $r(\ve x) = \rho(\ve x) / \int_{Q^\ell}\rho$,
  and let $\sigma^\ell(\de \ve x) = \ell^{-2} \de \ve x$
  be the uniform probability measure on $Q^\ell$.
  Let  us denote with  $\{\ve x_i\}_{i=1}^R$ a sample of $R$
  points independently distributed with $\nu$, 
  and with  $\{\ve z_i\}_{i=1}^R$ a sample of  $N$
  points independently distributed with the uniform probability measure
  $\sigma^\ell$, and let us indicate with
  $X^R(\de \ve x)$ and $Z^R(\de \ve x)$
  the corresponding empirical measures.
  
  Then there exists a constant
  $c>0$ such that for sufficiently small $\ell$
  \begin{equation*}
    \mathbb E_{\nu} W_2^2(X^R,\nu) \le (1+c\ell)
    \mathbb E_{\sigma^\ell} W_2^2(Z^R,\sigma^\ell).
  \end{equation*}
\end{lemma}
\begin{proof}
  Let us denote with $L$ the Lipschitz constant of $\rho$,
  and with $a$ a constant such that $\rho(\ve x)\ge a > 0$.
  The measure $\nu$ is approximated by $\sigma^\ell$ in the
  sense that
  $$\left|r(\ve x) - \frac 1{\ell^2}\right| =
  \frac 1{\int_{Q^\ell} \rho}
  \left| \rho(\ve x) -\frac 1{\ell^2}\int_{Q^\ell} \rho\ 
  \right| \le \frac L{a\ell}.$$
  Moreover,
  $$|r(\ve x) - r(\ve y)| \le \frac {L}{a\ell^2}.$$
  Let us define
  $$r_2(x_2)  = \int_0^{\ell} r(x_1',x_2)\de x_1'$$
  and note that
  $$\left|r_2(x_2) - \frac 1{\ell}\right| \le \frac La.$$
  We consider the map
  $$
  \left\{
  \begin{aligned}
    & G_1(x_1,x_2 ) = \ell \frac 1{r_2(x_2)}
    \int_0^{x_1} r(x_1',x_2) \de x_1'\\
    & G_2(x_1,x_2) = \ell \int_0^{x_2} r_2(x_2') \de x_2'.
  \end{aligned}
  \right.
  $$
  The map $\ve x=(x_1,x_2) \to \ve G = (G_1,G_2)$ is continuously
  differentiable, 
  its Jacobian is $r(\ve x)$, and it is bijective 
  from $Q^\ell$ in  $Q^\ell$.
  Then, if $\ve x$ is uniformly distributed on $Q^\ell$, $\ve G(\ve x)$
  is distributed with density $r$. The inverse map 
  $\ve \Gamma$ of $\ve G$ transports the uniform
  distribution $\sigma^\ell(\de \ve x)$ in the probability measure
  $\nu(\de \ve x)$ of density $r$.
  By definition of $\ve \Gamma$
  $$W_2^2(X^R,\nu) =
  \inf_{J} \int J(\de \ve x,\de \ve y)
  |\ve \Gamma(\ve x) - \ve \Gamma(\ve y)|^2,
  $$
  where the infimum in taken on the joint probability measures of
  $Z^n(\de \ve x)$ and $\sigma^\ell
  (\de \ve y)$, with $\ve z_i=\ve G(\ve x_i)$.
  Now we show that
  $$
  \begin{aligned}
  |\ve \Gamma(\ve  x) - \ve \Gamma(\ve y)|^2
  &\le |\ve x-\ve y|^2 \sup_{\ve x\neq \ve y}
  \frac {|\ve \Gamma(\ve x) - \ve \Gamma(\ve y)|^2}{|\ve x - \ve y|^2}\\
  &= |\ve x-\ve y|^2
  \sup_{\ve x\neq \ve y}
    \frac {|\ve x - \ve y|^2}{|\ve G(\ve x) - \ve G(\ve y)|^2} \le
    (1 + c\ell)|\ve x - \ve y|^2
  \end{aligned}
  $$
  from which the proof follows immediately.
  Let us define
  $$
  \begin{aligned}
    &\alpha
    =\frac \ell{x_2 -y_2}
    \int_{[x_2.y_2]} r_2(x_2')\de x_2',\\
    &\beta=
    \frac {\ell}{x_1 - y_1}\int_{[x_1.y_1]} \frac {r(x_1',x_2)}{r_2(x_2)}
      \de  x_1',\\
    &\gamma=
    \frac \ell{x_2 - y_2}  \int_0^{y_1} \left(\frac {r(x_1',x_2)}{r_2(x_2)} -
      \frac {r(x_1',y_2)}{r_2(y_2)}\right) \de x_1'.
  \end{aligned}
  $$
  Using the estimate on $r-1/\ell^2$, $r_2-1/\ell$ and
  on the Lipschitz constant of $r$, we have
  $$
  \begin{aligned}
    &|\alpha -1 | \le c \ell, &
    &|\beta -1 | \le c \ell, &
    &|\gamma| \le c\ell.
  \end{aligned}
  $$
  Then
  \begin{equation*}
    \begin{aligned}
    |\ve G(\ve x) - \ve G(\ve y)|^2 &=
    (\alpha^2 + \gamma^2) (x_2 - y_2)^2 + \beta^2 (x_1 - y_1)^2
    + 2 \beta \gamma (x_1 - y_1) (x_2-y_2) \\
    &\ge (1-c\ell) |\ve x - \ve y|^2,
    \end{aligned}
  \end{equation*}
  for a suitable constant $c$ and $\ell$ sufficiently small.
\end{proof}

\vskip.3cm
We will also need to 
bound the $2-$Wasserstein distance
between two slightly different and positive densities on the square.
We can do this with the following Lemma, 
which is a corollary  of Benamou-Brenier formula \cite{BB}.

\begin{lemma}
  \label{lemma:wl2}
  If $\nu_1$ and $\nu_2$  are two
  probability measures on a convex domain $\Lambda$, 
  absolutely continuous with respect to the Lebesgue measure,
  with densities bounded from below and from above by finite  
  non-zero constants, then
  $$W_2^2(\nu_1,\nu_2) \le c \|\nu_1 - \nu_2\|_2^2.$$
\end{lemma}

\begin{proof}
  The Benamou-Brenier formula allows to estimate the
  $2-$Wasserstein distance between two measures in terms of
  the $\dot {\mathbb H}^{-1}$ norm of their difference. 
  More precisely, Theorem 5.34 in \cite{Sant} says:
  if $\nu_1$ and $\nu_2$ are two absolutely continuous
  measures defined on a convex domain 
  $\Lambda,$ with densities bounded from below and from above by
  the constants $a$ and $b$ respectively, $0<a<b$, then 
  \begin{equation*}
    \frac{1}{\sqrt{b}}
    \| \nu_1 - \nu_2 \|_{\dot{H}^{-1}(\Lambda)}\leq
    W_2(\nu_1,\nu_2) \leq \frac{1}{\sqrt{a}}
    \| \nu_1 - \nu_2 \|_{\dot{H}^{-1}(\Lambda)},
\end{equation*}
where the ${\dot H}^{-1}$ norm of a $0-$average charge distribution $\nu$ 
is defined by 
$$\| \nu \|_{\dot{H}^{-1}(\Lambda)}=
\int_{\Lambda}|\nabla
\lap^{-1}\nu|^2,$$
where the inverse of Laplacian is defined with Neumann
homogeneous 
boundary conditions on $\pa\Lambda$.
Therefore, by noticing that the $\dot{H}^{-1}$  norm 
is bounded from above by a positive constant
depending only on $|\Lambda|$ times the $L_2$ norm,
we get the result.
\end{proof}
We remark that more general results,
including the case of non-convex domains,
can be found in \cite{Pey} and
references therein. We also remark that this Lemma fails
if the supports of the measures are not connected,
according to remark \ref{rem:delta} at the end of the previous
section.

\vskip.3cm
Now we can start to prove Theorem
\ref{teo:stima}.
Let $m$ a positive integer and 
let us cover $Q=[0,1]^2$ 
with the $m^2$ squares $\{Q_k\}_{k=1}^{m^2}$,
of sides $1/m$ and of measure $1/m^2$,
given by 
$[i/m,(i+1)/m]\times[j/m,(j+1)/m]$, with $i,j = 0,\dots m-1$.
as in fig.
\ref{fig:moltiquadrati}.

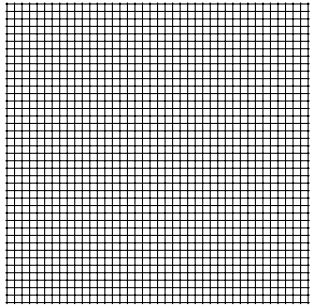
\begin{figure}[hbt!]
\centering
\begin{tikzpicture}
\draw[step=0.099cm] (-2,-2) grid (2,2);
\end{tikzpicture}
\caption{Regular grid of squares.}
\label{fig:moltiquadrati}
\end{figure}

We define:
\begin{align}
  &\sigma_k(\de \ve x) = m^2 \de \ve x \text{ the 
    uniform probability measure on }Q_k\\
  \label{eq:pk}
  &p_k = \int_{Q_k} \rho(\ve x) \de \ve x \text{ 
  the probability that $\ve x$, extracted with $\mu$,
    belongs to $Q_k$}\\
    \label{eq:mukm}
  &\mu_k^m(\de \ve x) = \frac 1{p_k} \rho(\ve x) \de \ve x
    \text{ the distribution of $\ve x$, conditioned to }\ve x \in Q_k.
\end{align}
Let $\{\ve x_i\}_{i=1}^N$ be
a sample of $N$ independent points distributed with $\mu$,
and let us denote with $R_k$ the number of points $\ve x_i$
in the square $Q_k$.
Let $J_k(\de \ve x,\de \ve y)$ a joint probability distribution
on $Q_k\times Q_k$
with marginals given by
$$
\begin{aligned}
  &\int_{Q_k} J_k(\de \ve x, \cdot ) = X^N_k(\de \ve x)
  := \frac 1R_k \sum_{j=1}^N\fcr\{\ve x_j \in Q_k\}
  \delta_{\ve x_j}(\ve x) \de \ve x\\
  &\int_{Q_k} J_k(\cdot, \de \ve y ) = \mu_k^m(\de \ve y).
\end{aligned}
$$
Then
$$J(\de \ve x,\de \ve y) = \sum_{k=1}^{m^2} \frac {R_k}N
J_k(\de \ve x,\de \ve y)$$
is a joint distribution in $Q\times Q$ with marginals given by
$$
\begin{aligned}
  &\int_Q J(\de \ve x, \cdot ) = X^N(\de \ve x)
  = \frac 1N \sum_{j=1}^N \delta_{\ve x_j}(\ve x) \de \ve x\\
  &\int_Q J(\cdot, \de \ve y ) = \mu^m(\de \ve y) :=
  \sum_{k=1}^{m^2} \frac {R_k}N \mu^m_k (\de \ve y)=
  \sum_{k=1}^{m^2} \frac {R_k}{p_kN} \rho(\ve y)
  \fcr\{\ve y\in Q_k\} \de \ve y.
\end{aligned}
$$
We will estimate $\mathbb E[W^2_2(X^N,\mu)]$ by the triangular
inequality, 
trough the estimates of $\mathbb E[W^2_2(X^N,\mu^m)]$ and
$\mathbb E[W^2_2(\mu^m,\mu)]$.

\subsection*{Estimate of  $\mathbb E[W^2_2(X^N,\mu^m)]$.}
\label{sub:1}

By definition
$$W_2^2(X^N,\mu^m)\le \sum_{i=1}^{m^2} \frac {R_k}N
W_2^2(X^N_k,\mu^m_k).$$
We first take the  expected value conditioned to the variables $R_k$,
which is equivalent to  fix $\{R_k\}$ and to extract a sample of
$R_k$ particle in $Q_k$ with distribution $\mu_k^m$,
as defined in \eqref{eq:mukm}.
Then we will take the expectation in $\{R_k\}$ with respect to $\mu$,
which means to extract the multinomial variables $\{R_k\}$
with probability $p_k$, as defined in \eqref{eq:pk}:
$$
\mathbb E_\mu
\left[
W_2^2(X^N,\mu^m)|\,\{R_k\}_{k=1}^{m^2}
\right]
\le  \sum_{i=1}^{m^2}
\frac {R_k}N
\mathbb E_{\mu_k^m}
W_2^2(X^N_k,\mu^m_k).
$$
We estimate  $\mathbb E_{\mu_k^m} W_2^2(X^N_k,\mu^m_k)$
using Lemma \ref{lemma:ell}, identifying $\ell = 1/m$,
$Q^\ell$ with $Q_k$, and using the results in Eq. \eqref{fact2}:
\begin{equation}
  \begin{aligned}
  \mathbb E_{\mu_k^m} [W^2_2(X^N_k,\mu_k^m)]
  &\le (1+c/m) \mathbb E_{\sigma_k}[W^2_2(Z^{R_k},\sigma_k)]\\
  &= (1+c/m) \frac 1{m^2}
  \left( \frac { \log R_k}{4\pi R_k} + o(\log R_k/R_k)\right).
  \end{aligned}
\end{equation}
Then, multiplying for $R_k/N$ and summing on $k$
$$\mathbb E_\mu\left[W_2^2(X^N,\mu^m)|\,\{R_k\}_{k=1}^{m^2}
\right]
\le  (1+c/m) \sum_{k:R_k>0}\left( \frac 1{m^2}
  \frac 1{4\pi N} \log R_k
  + \frac 1{m^2} o( \log R_k / N)\right).$$
The expected value of $R_k$ is $N_k=p_k N$,
where $p_k$ is of order $1/m^2$. Then we need
that $N/m^2$ diverges with $N$.
For $N$ large, $R_k$ 
differs from $N_k$ of a term of order $\sqrt N/m$,
then
$$\mathbb E \left[ \sum_{k:R_k>0} \frac 1{m^2}
  o( \log R_k / N)\right] = o(\log N/N) + o(\log m/N)
= o(\log N/N). 
$$
Moreover
$$\log R_K = \log (R_k/N_k) + \log p_k + \log N \le
\log (R_k/N_k) + \log N,$$
and since $p_k\le 1$ and
since $\log $ is a convex function 
$$\mathbb E_\mu [\log R_k] \le  \mathbb E_{\mu} [\log (R_k/N_k)]
+ \log N \le
\log N.
$$
Therefore, we conclude that
$$\mathbb E_\mu[W_2^2(X^N,\mu^m)]
\le  (1+c/m) \left( \frac {\log N}{4\pi N} + o(\log N/N)\right).$$

\subsection*{Estimate of  $\mathbb E[W^2_2(\mu^m,\mu)]$.}
\label{sub:2}

Here we use Lemma \ref{lemma:wl2}:
$$
\begin{aligned}
W_2^2(\mu^m,\mu) &\le c\|\mu^m-\mu\|^2_2
= c \sum_{k=1}^{m^2}
\left( \frac {R_k}{p_kN} - 1\right)^2 \int_{Q_k} \rho(\ve x)^2
\de \ve x  \\
&= c \sum_{k=1}^{m^2} \frac 1{p_k^2N^2} (R_k - p_k N)^2
\int_{Q_k} \rho(\ve x)^2\de \ve x.\end{aligned}
$$
Taking the expectation
$$
\begin{aligned}
  \mathbb E_\mu[W_2^2(\mu^m,\mu) &\le
c \sum_{k=1}^{m^2} \frac 1{p_k^2N^2} Np_k(1-p_k)
\int_{Q_k} \rho(\ve x)^2\de \ve x  \\
& \le 
c \frac 1N \sum_{k=1}^{m^2} \frac 1{p_k}
\int_{Q_k} \rho(\ve x)^2\de \ve x  \le c \|\rho\|_{\infty}
\frac {m^2}N.
\end{aligned}
$$

\vskip.3cm
\begin{proof}[Proof of Theorem \ref{teo:stima}.]
   Using that
  $(a+b)^2 \le (1+\delta) a^2 + (1+1/\delta) b^2$ for any $\delta >0$,
  from the triangular inequality for $W_2$ 
  we have
$$
\begin{aligned}
  \frac N{\log N} \mathbb E_\mu [W^2_2(X^n,\mu)] \le &
  (1+\delta ) \frac N{\log N}\mathbb E_\mu [W^2_2(X^n,\mu^m)] \\
  &+
  (1+1/\delta)\frac N{\log N}  \mathbb E_\mu [W^2_2(\mu^m,\mu)] \\
  \le & (1 + \delta ) (1+c/m)\left(
    \frac 1{4\pi} + o(1)\right) + 
  c(1+1/\delta) \frac {m^2}{\log N}. 
\end{aligned}
$$
We achieve the proof of  eq. \eqref{estimate2} 
taking the $\limsup$ in $N$
and then passing to the limit $m\to + \infty$ and $\delta \to 0$.
\vskip.3cm

To prove estimate \eqref{estimate1} we use a nice argument introduced in
\cite[Prop. 2.1]{AST}.
For first, let us remind that  
the best coupling between an absolute
continue measure $\mu$ and $X^N$
can be represented with a
measurable map
$\ve T_{\scriptscriptstyle{X^N}}:Q\to Q$ such that
$\ve T_{\scriptscriptstyle{X^N}}$ transport $\mu(\de \ve x)$
in $X^N(\de \ve x)$, and 
$$J_T(\de \ve x,\de \ve y)=
\delta(\ve y -  \ve T_{\scriptscriptstyle{X^N}}(\ve x))\rho(\ve x) \de \ve x \, \de \ve y$$
is the joint distribution which realize the infimum in the
definition of the 2-Wasserstein distance:
$$W_2^2(\mu,X^N) = \int J_T(\de \ve x, \de \ve y) |\ve x- \ve y|^2 = 
\int |T_{\scriptscriptstyle{X^N}}(\ve x)-\ve x|^2 \rho(\ve x) \de \ve x.$$
Let $Y^N$ another empirical  measure obtained extracting $N$ particles
with  distribution $\mu$, and let $\ve T_{\scriptscriptstyle{Y^N}}$ be
the
corresponding map which
gives the best coupling.
Then, since $\ve T_{\scriptscriptstyle{X^N}}$ and $\ve T_{\scriptscriptstyle{Y^N}}$ transport $\mu$ in $X^N$ and $Y^N$
respectively,
\begin{equation*}
  \begin{split}
    W_2^2(X^N,Y^N)&
    \leq\int\,|\ve T_{\scriptscriptstyle{X^N}}(\ve x)-\ve T_{\scriptscriptstyle{Y^N}}(\ve x))|^2
    \rho(\ve x) \de \ve x\\
    &= \int
    |\ve T_{\scriptscriptstyle{X^N}}(\ve x) - \ve x - (\ve T_{\scriptscriptstyle{Y^N}}(\ve x) - \ve x)|^2
    \rho(\ve x) \de \ve x\\ 
    &=
    \int  |\ve T_{\scriptscriptstyle{X^N}}(\ve x)-\ve x|^2\rho(\ve x) \, \de \ve x+
    \int  |\ve T_{\scriptscriptstyle{Y^N}}(\ve x)-\ve x|^2\rho(\ve x) \, \de \ve x+\\
    &-2\int
    (\ve T_{\scriptscriptstyle{X^N}}(\ve x)-\ve x)\cdot
    (\ve T_{\scriptscriptstyle{Y^N}}(\ve x)-\ve x)\rho(\ve x) \de \ve x.
\end{split}
\end{equation*}
Considering that, since $X^N$ and $Y^N$ are independent and identically
distributed, 
also $\ve T_{\scriptscriptstyle{X^N}}(\ve x)$ and
$\ve T_{\scriptscriptstyle{Y^N}}(\ve x)$ are independent
and identically distributed. Then, taking the expectation,  
\begin{equation*}
\begin{split}
  \mathbb{E}_\mu[W_2^2(X^N,Y^N)]&\leq 2 \mathbb{E}_\mu
  \int |\ve T_{\scriptscriptstyle{X^N}}(\ve x)-\ve x|^2\rho(\ve x)  \de \ve x  -
  2\int |\mathbb{E}_\mu [\ve T_{\scriptscriptstyle{X^N}}(\ve x) -\ve x]|^2\rho(\ve x) \de \ve x\\
&\leq 2 \mathbb{E}_\mu
\int |\ve T_{\scriptscriptstyle{X^N}}(\ve x)-\ve x|^2\rho(\ve x) \de \ve x =
\mathbb E_{\mu} W_2^2(X^N,\mu).
\end{split}
\end{equation*}
Therefore, by \eqref{estimate2} we get \eqref{estimate1}. 
\end{proof}

\section*{Acknowledgments}

We thank Luigi Orsina and Riccardo Salvati Manni for useful
discussions and suggestions.

This work has been partially supported by the grant
``Progetti di ricerca d'Ateneo 2016''  by Sapienza University, Rome.

\end{document}